\documentclass{article}
\usepackage[margin=1in]{geometry}

\usepackage{amsmath}
\usepackage{amssymb}
\usepackage{amsthm}
\usepackage{booktabs}
\usepackage{graphicx}
\usepackage{hyperref}
\usepackage{natbib}

\newcommand{\weakly}{\xrightarrow{d}}
\newcommand{\distequal}{\overset{d}{=}}
\newcommand{\inprob}{\overset{P}{\to}}
\newcommand{\inas}{\overset{a.s.}{\to}}
\newcommand{\expit}{\textnormal{expit}}
\newcommand{\est}{\hat{\Psi}_n}
\newcommand{\bootest}{\hat{\Psi}^*_m}
\newcommand{\data}{\mathcal{D}_n}

\newcommand{\csint}{\mathcal{I}_{(m,n,B)}}

\newcommand{\estimand}{\Psi(P)}
\newcommand{\normal}{\mathcal{N}(0, \sigma^2)}

\newcommand{\subsamplingfactor}{\sqrt{\frac{mn}{n-m}}}
\newcommand{\subpercentage}{\frac{m}{n}}

\title{Cheap subsampling bootstrap confidence intervals with applications to longitudinal causal inference}

\author{\small Johan Sebastian Ohlendorff$^{1,2,*}$, Anders Munch$^{1}$,  Kathrine Kold Sørensen$^{2}$, \\
\small and Thomas Alexander Gerds$^{1,2}$ \\
\small $^{1}$Section of Biostatistics, University of Copenhagen, Denmark\\ \small$^{2}$Department of Clinical and Translational Research, Steno Diabetes Center Copenhagen, Denmark}

\newtheorem{theorem}{Theorem}

\begin{document}

\maketitle

\begin{abstract}
  The bootstrap is widely used to construct confidence intervals
  but may perform poorly for some estimators.
  We consider semiparametric estimators that use cross-validation
  to estimate nuisance parameters and can be computationally expensive
  in large datasets. We illustrate potential problems
  with bootstrapping with replacement when cross-validation is used
  and extend the cheap bootstrap confidence interval to
  subsampling without replacement for asymptotically linear estimators.
  We establish asymptotically valid coverage with few bootstrap
  replications when the subsample size is chosen appropriately
  as a function of the sample size. Empirical experiments demonstrate
  computational advantages over resampling with replacement
  and explore the impact of
  subsample size, sample size, cross-validation and the number of
  bootstrap repetitions on the performance of the confidence interval.
  We apply the
  method to data from the LEADER trial to construct confidence intervals
  for a longitudinal targeted minimum loss-based estimator.
\end{abstract}

{\it Keywords:} bootstrap; causal inference; subsampling; targeted learning.

\section{Introduction}

The bootstrap is a standard approach to obtain confidence intervals in
cases where estimators of standard errors do not exist or perform
poorly. Motivated by the prohibitive cost of repeated model fitting in
conventional bootstrap inference, \citet{lamCheapBootstrapMethod2022}
introduced the cheap bootstrap, which provides asymptotically exact
confidence intervals under standard bootstrap regularity conditions
using any fixed number of resamples. In this article, we build on
Lam's seminal work and derive a cheap bootstrap confidence interval
based on subsampling \citep{wuJackknife1990,
  politisSubsampling1999}. Our algorithm subsamples bootstrap datasets
without replacement that are smaller than the original dataset. We are
motivated by epidemiological studies of observational data where large
sample sizes are frequent and estimators of target parameters use
machine learning to learn nuisance parameter models
\citep{hernan2016using,hernan2020whatif,Chernozhukov2018,vanderlaanTargetedLearningData2018}.
We show that when bootstrap samples are drawn with replacement,
complications may arise when the learners of the nuisance parameters
use bootstrap, cross-validation, or cross-fitting. If we split a
bootstrap sample with duplicate observations into training and testing
folds, the observations from the original dataset are included
multiple times in both the training and the testing fold (see
Figure~\ref{fig:bootstrap_problem}).
\begin{figure}[htbp]
\centering
\includegraphics[width=.65\textwidth]{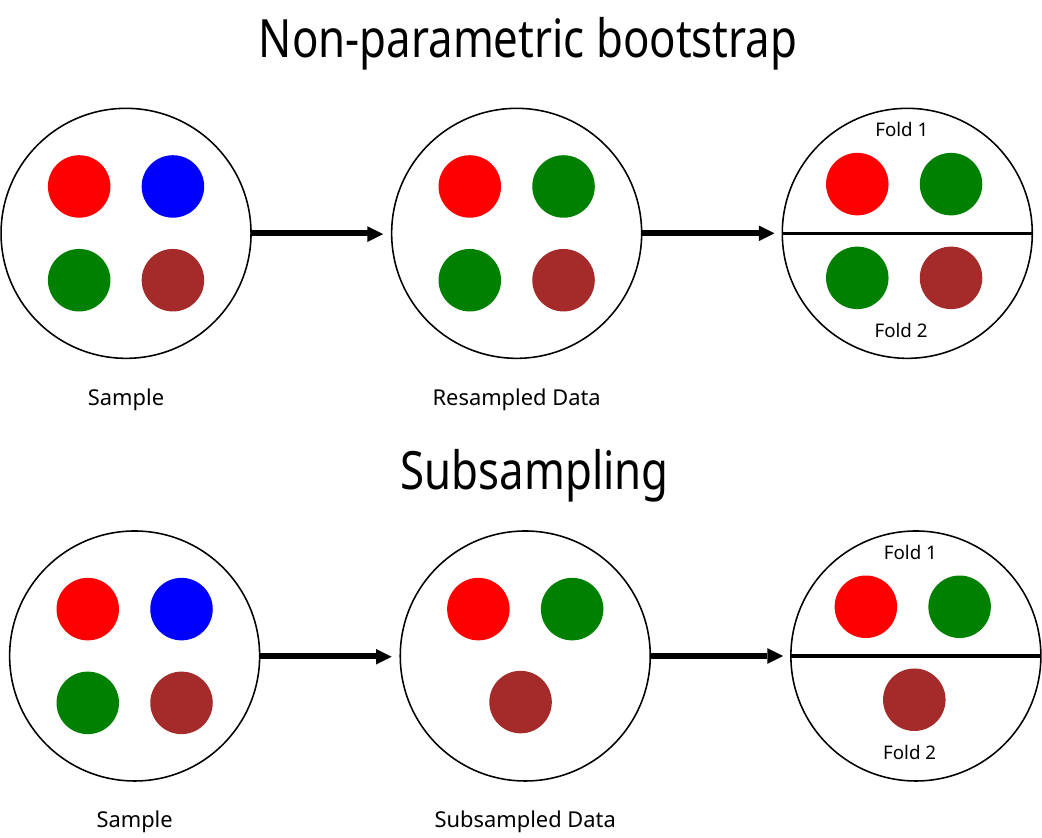}
\caption[Illustration of the problem with ties when using bootstrap with replacement]{\label{fig:bootstrap_problem}Illustration of the problem with ties in the data when using bootstrap with replacement and the subsampling bootstrap. First, a bootstrap sample is drawn with and without replacement from the dataset. The dataset is then split into two folds for the cross-validation procedure. We see that there is one observation present in both folds for the bootstrap sample drawn with replacement. Conversely, the subsample does not have this issue.}
\end{figure}
To illustrate that this violation of the independence
between the training and test sets is a real concern, we simulate data
(\(n=1,000\)) from the model \(P(Y=1|X=x) = \operatorname{expit} (x^2)\)
where \(X\) is a continuous uniformly distributed random
variable. Figure~\ref{fig:bootstrap_problem_example} shows the data
generating model and a random forest prediction model
\citep{breiman2001random} based on 1,000 trees and a minimum node size
of 300. The forest severely overfits the data. However, the 10-fold
cross-validated prediction performance based on a non-parametric
bootstrap sample drawn with replacement misleadingly indicates that the forest is
outperforming a correctly specified logistic regression model (10-fold CV-AUC: forest
64.8\%, data generating model 56.9\%, 10-fold CV Brier score: forest
22.8\%, data generating model 23.9\%). The Brier score
is a proper scoring rule and this implies that the data
generating model necessarily has the best prediction performance
on average
\citep{gneiting2007strictly}. Sampling without replacement does not
have this problem. In a subsample of 632 observations, 10-fold
cross-validation shows that the data generating model has better
prediction performance (10-fold CV-AUC: forest 55.1\%, data generating
model 56.1\%, 10-fold CV Brier score: forest 24.0\%, data generating
model 23.7\%). 

\begin{figure}[htbp]
\centering
\includegraphics[width=.7\textwidth]{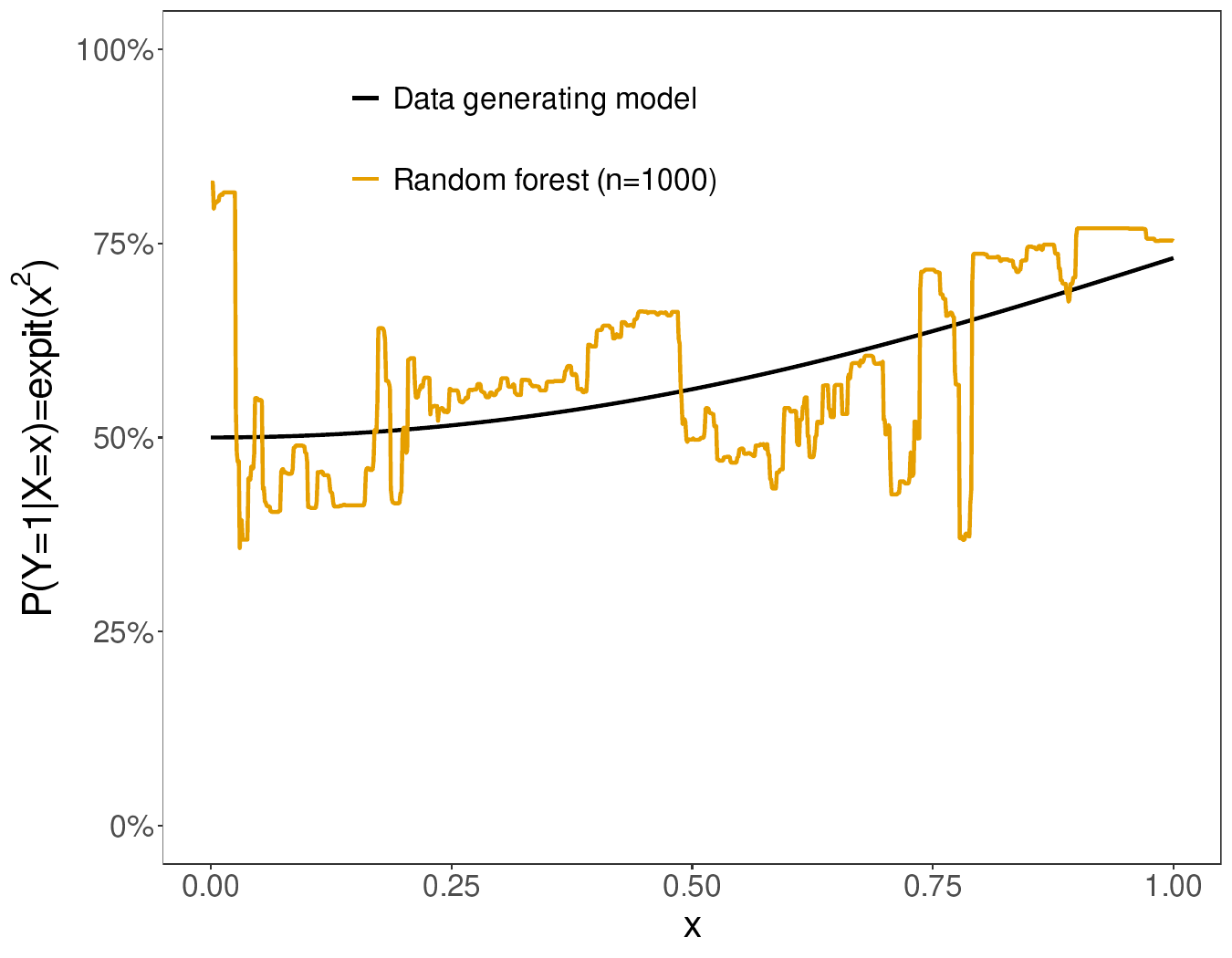}
\caption[Illustration of an overfitted random forest and the data generating model]{\label{fig:bootstrap_problem_example} The figure shows the
data generating model \(P(Y=1|X=x) = \operatorname{expit({x^2})}\) for
a continuous uniformly distributed random variable \(X \sim \text{Unif}(0,1)\) and the predicted
probabilities of a random forest fitted to a sample with 1000
observations using 1000 trees and a minimum node size of 300.
Evaluating the prediction performance using 10-fold cross-validation
on a non-parametric bootstrap sample indicates wrongly that the forest
outperforms a correctly specified logistic regression model (higher AUC and lower Brier
score).}
\end{figure}

\cite{lamCheapBootstrapMethod2022} also considers subsampling
methods such as the cheap \(m\)-out-of-\(n\) bootstrap, the cheap
subsampled double bootstrap, and the cheap bag of little bootstraps.
However, all of these approaches involve drawing bootstrap samples
with replacement at some point, either from the full dataset or from a
subset of the data.
These approaches require modifications
when combined with cross-validation,
because naively dividing bootstrap samples can place copies of the
same original observation in both training and validation sets.
Modified cross-validation schemes can prevent this by
keeping track of the duplicate observations
\citep{jiang2008, cai2025bootstrapping}.
Subsampling without replacement avoids this complication
without requiring such modifications.

Building on \cite{lamCheapBootstrapMethod2022} and using a result from
\cite{wuJackknife1990}, we establish the asymptotic validity of our
cheap subsampling confidence interval for asymptotically linear
estimators under the condition that the subsample size \(m\) remains
bounded away from the full sample size \(n\). To improve finite-sample
performance we use a scaling factor \(\sqrt{\frac{m}{n-m}}\) which is
motivated by similarity with the delete-\(d\) jackknife variance
estimator \citep{jackknifegeneraltheory} and the finite-sample
correction factor for subsampling \citep{politisSubsampling1999}.
Like Lam's cheap bootstrap, our method uses quantiles from a
\(t\)-distribution and in theory works with a small number of
subsamples. The resulting computational efficiency gain, however,
comes at the cost of increased Monte Carlo error and increased width
of the confidence intervals when the number of subsamples is very low.

We evaluate the finite-sample coverage and width of cheap subsampling
confidence intervals, including their sensitivity to the subsample
size and number of resamples, through simulations and a reanalysis of
the LEADER trial \citep{sorensen2026accounting}, where the causal
\(t\)-year risk is estimated using longitudinal targeted minimum
loss-based estimation \citep{laanTargetedMinimumLoss2012}. We also
compare our method with Lam's cheap bootstrap to investigate the use
of cross-validation with resampling with replacement. Section
\ref{cheap_sub} presents the method, its asymptotic validity, and its
connection to the delete-\(d\) jackknife variance estimator, followed
by the LEADER analysis in Section \ref{example_ltmle} and the
simulation study in Section \ref{sim_ltmle}.

\section{Cheap subsampling bootstrap}
\label{cheap_sub}
We introduce our cheap subsampling algorithm in a
general framework that includes the applied settings. Let \(\data =
(O_1, \dots, O_n)\) with \(O_i \in \mathbb{R}^d\) be a dataset of
independent and identically distributed random variables sampled from
some unknown probability measure \(P \in \mathcal{P}\). Here
\(\mathcal{P}\) denotes a suitably large set of probability measures
on \(\mathbb{R}^d\). We denote by \(\Psi : \mathcal{P} \rightarrow
\mathbb{R}\) a statistical functional of interest and by \(\est\) an
estimator of \(\estimand\) based on \(\data\).  The estimator
\(\hat{\Psi}_n\) is asymptotically linear \citep{bickel1993efficient}
if
\begin{equation}
\est-\estimand = \frac{1}{n} \sum_{i=1}^n \phi_P(O_i) + R_n(P)
\end{equation}
where \(\phi_P : \mathbb{R}^d \rightarrow \mathbb{R}\) is a measurable
function with \(\mathbb{E}_P[\phi_P(O)] = 0\),
\(0<\mathbb{E}_P[\phi_P(O)^2] < \infty\), and the remainder term
fulfills \(R_n(P) = o_P(\frac{1}{\sqrt{n}})\) for all \(P \in
\mathcal{P}\). A subsample \(\mathcal{D}^*_m= (O_1^*, \dots, O_m^*)\)
is a diminished dataset obtained by drawing \(m < n\) observations
without replacement from the dataset \(\data\). We denote by
\(\bootest\) the estimate based on the subsample \(\mathcal{D}^*_m\).

Our theoretical results are based on Lam's seminal work
\citep{lamCheapBootstrapMethod2022} on the cheap bootstrap. 
We construct a cheap subsampling confidence interval for \(\Psi(P)\)
based on the estimator \(\est\) and \(B\geq 1\) subsamples of size
\(m\). By repeating the subsampling procedure
independently \(B \geq 1\) times we obtain the estimates based on
subsamples \(\{\hat{\Psi}_{(m,1)}^*, \dots, \hat{\Psi}_{(m,B)}^*\}\)
and define the cheap subsampling confidence interval as
\begin{equation} \csint = \left(\est-t_{B,1-\alpha/2}
\sqrt{\frac{m}{n-m}} S, \est+t_{B,1-\alpha/2} \sqrt{\frac{m}{n-m}} S
\right), \label{eq:cs_ci}
\end{equation} where \(S= \sqrt{\frac{1}{B} \sum_{b=1}^B
\left(\hat{\Psi}_{(m,b)}^*-\est \right)^2}\) and \(t_{B,1-\alpha/2}\)
is the \(1-\alpha/2\) quantile of a \(t\)-distribution with \(B\)
degrees of freedom.

We show asymptotic validity of the cheap subsampling confidence interval
\eqref{eq:cs_ci} under the assumption that \(\est\) is
asymptotically linear and that the subsample size \(m(n)\) is chosen
appropriately as a function of the sample size \(n\).  To not overburden the
notation, we will continue to write \(m\) instead of \(m(n)\) in what
follows. Sufficient conditions for asymptotic validity are formulated
in Theorem \ref{thm:cs_validity}.

\begin{theorem}
Let \(\est\) be an asymptotically linear estimator of \(\estimand\).
If the subsample size \(m\) satisfies \(\sup_{n \in \mathbb{N}}
\subpercentage < c\) for some \(0<c<1\), then the coverage
probability of the cheap subsampling confidence interval \(\csint\)
converges to \(1-\alpha\), that is \(P(\estimand \in \csint)
\rightarrow 1-\alpha\), as \(m,n \rightarrow \infty\) for any fixed \(B \geq
1\).
\label{thm:cs_validity}
\end{theorem}

\begin{proof}
  The proof is given in Appendix \ref{pf_thm1}.
\end{proof}

Next, we state a result that shows that the endpoints of the cheap subsampling confidence interval converge to
a random limit fully determined by the data \(\data\)
corresponding to an asymptotic confidence interval based on the delete-\((n-m)\) jackknife variance estimator
\citep{jackknifegeneraltheory}
as the number of bootstrap repetitions increases.
This result is a direct consequence of the conditional law of large numbers (Theorem 4.2 of \cite{majerek2005conditional}),
which shows that if \(\est\) has a finite second moment,
then
\(
S^2 =  {\frac{1}{B} \sum_{b=1}^B
\left(\hat{\Psi}_{(m,b)}^*-\est \right)^2} \inas \mathbb{E}_P[(\hat{\Psi}^*_m-\est)^2| \data], 
\)
as \(B \rightarrow \infty\) for fixed \(m\) and \(n\).
Thus, we specifically have 
\begin{equation}
  \est \pm t_{B,1-\alpha/2} \sqrt{\frac{m}{n-m}} S \inas \est \pm q_{1-\alpha/2} \sqrt{\hat{V}_{\textnormal{jack}}}, \text{ as } B \rightarrow \infty,
\end{equation} 
where \(q_{1-\alpha/2}\) is the \(1-\alpha/2\) quantile of the standard normal distribution
and the delete-\((n-m)\) jackknife variance estimator is given by \(
\hat{V}_{\textnormal{jack}} = \frac{m}{n-m} \mathbb{E}_P[(\hat{\Psi}^*_m-\est)^2| \data]\).

\section{Application to the LEADER dataset}
\label{example_ltmle}

For the sole purpose of illustrating our method, we use the
data and setup from
a reanalysis \citep{sorensen2026accounting} of
the LEADER trial \citep{LEADER}, where we apply the cheap subsampling
confidence interval to the longitudinal targeted minimum loss-based
estimator (LTMLE) \citep{vanderlaanTargetedLearningData2018}.
We use a
subset of the data from the LEADER trial, which includes 8652 patients
with type 2 diabetes. The baseline variables selected for analysis
were sex, age, use of thiazide, use of statin, hypertension, BMI, and
the number of years since diabetes diagnosis. The time-varying data
were discretized into three 12-month intervals. In each
interval, we defined time-varying covariates that represent HbA1c, use
of thiazolidinediones, use of sulfonylureas, use of metformin, and use
of DPP-4 inhibitors.
For each time interval, the LTMLE algorithm
estimates nuisance regression models for the outcome, the propensity
of treatment, and the censoring probability.
10-fold cross-validation was used to select
between lasso and ridge regression \citep{glmnet, elasticnet}
for the estimation of the nuisance parameters for the LTMLE procedure.
Throughout, we apply the cheap subsampling algorithm to obtain
confidence intervals for the LTMLE estimate of the absolute 3-year risk of all-cause death
under complete adherence to the placebo regimen.
We compare our method
with asymptotic confidence intervals that are based on an estimate of
the efficient influence function
\citep{laanTargetedMinimumLoss2012}. We investigate the Monte Carlo
error (effect of setting the random seed), the impact of the subsample
size, and the number of bootstrap repetitions on the cheap subsampling
confidence interval.

Figure~\ref{fig:confidence_interval_leader} shows the effect of the
number of bootstrap repetitions \(B\) on the cheap subsampling
confidence interval. While the cheap subsampling confidence intervals are
initially very wide, they stabilize quickly and do not change
significantly after 10 bootstrap replications.

Due to the random nature of the bootstrap, the bootstrap confidence
intervals are affected by Monte Carlo error, i.e., they depend on the
random seed set by the algorithm. This entails that the confidence
intervals will change when repeating the whole bootstrap procedure
with a different random seed. In Figure~\ref{fig:seed}, we illustrate this seed
effect on the upper endpoint of the confidence interval for various
subsample sizes and 10 repetitions of the whole cheap subsampling
bootstrap algorithm. In this particular application, \(B\geq 25\)
reduces a significant amount of the Monte Carlo error in our data example.
We also see that
in this example the cheap subsampling confidence interval does not
depend on the subsample size with sufficiently many bootstrap
repetitions (\(B \geq 25\)).

\begin{figure*}[!t]
\centering
\includegraphics[width=.65\textwidth]{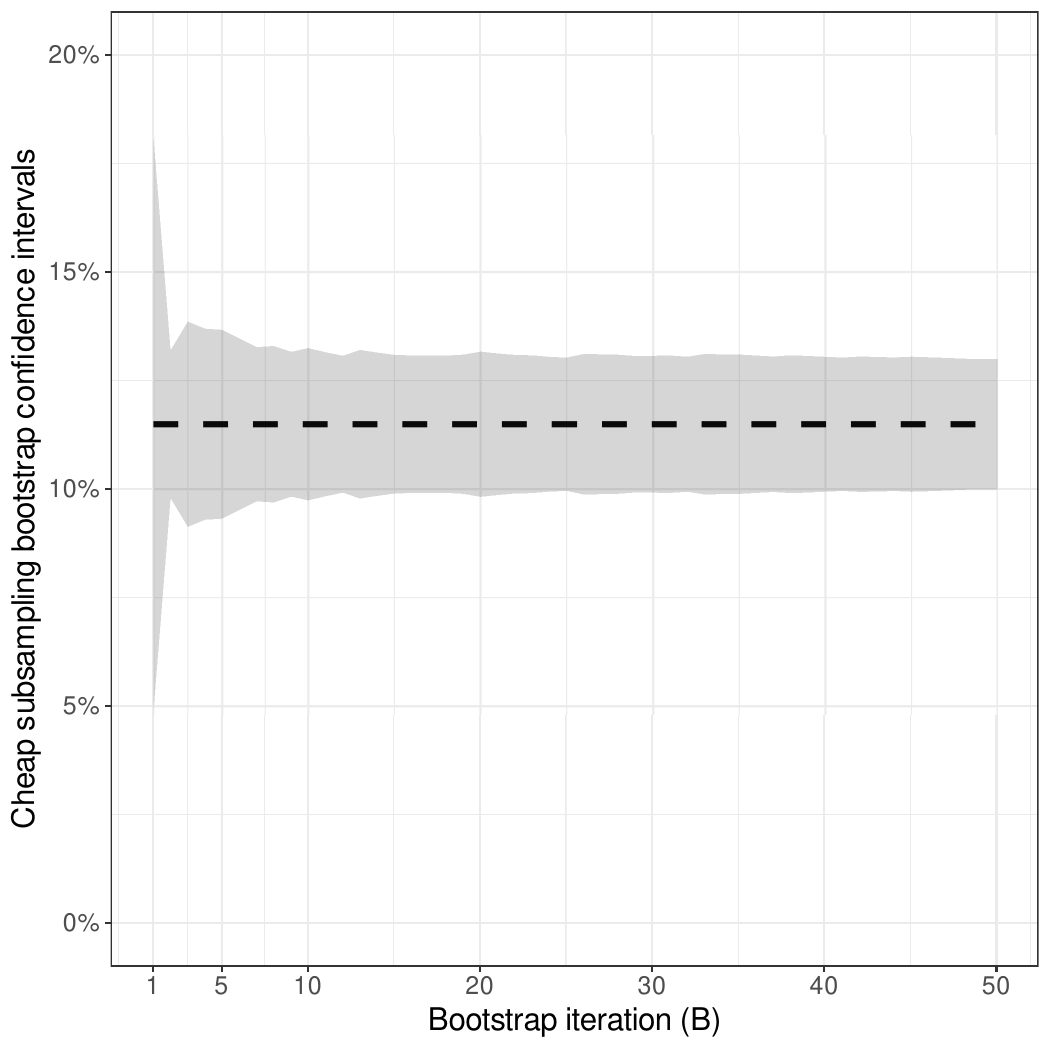}
\caption[95\% cheap subsampling confidence intervals for the absolute risk in the LEADER trial against number of bootstrap replications]{\label{fig:confidence_interval_leader}95\% cheap subsampling confidence intervals for the LTMLE estimate of the absolute risk of dying within 3 years under the placebo regimen in the LEADER trial. The x-axis shows the number of bootstrap repetitions \(B \in \{1, \dots, 50\}\) for the subsample size \(m = \lfloor 0.8 \cdot 8652 \rfloor = 6921\). The point estimate is displayed as the dashed line. }
\end{figure*}

\begin{figure*}[t]
\centering
\includegraphics[width=.65\textwidth]{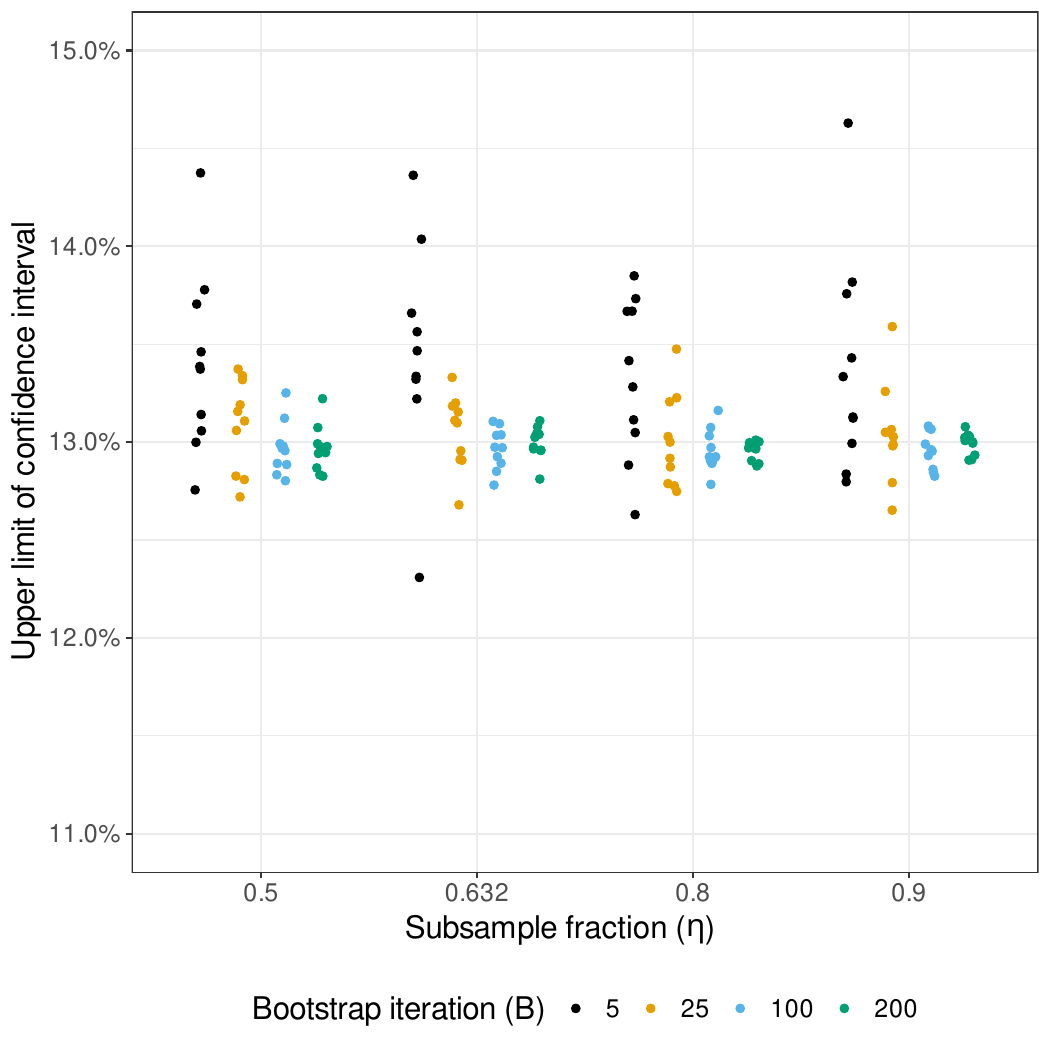}
\caption [Upper endpoint of 95\% cheap subsampling confidence interval of the absolute risk in the LEADER trial against subsample size across repetitions]{\label{fig:seed}The upper endpoints of 95\% cheap subsampling confidence intervals for the LTMLE estimate of the absolute risk of dying within 3 years under the placebo regimen in the LEADER trial. The plot shows the Monte Carlo error (random seed effect) based on 10 runs of the cheap subsampling algorithm for each of the subsample sizes \(m = \lfloor \eta \cdot 8652 \rfloor\) with \(\eta \in \{0.5, 0.632, 0.8, 0.9\}\) and number of bootstrap repetitions \(B \in \{5, 25,100,200\}\).}
\end{figure*}

\section{Simulation study}
\label{sim_ltmle} In this section, we simulate data to investigate the
effects of sample size, subsample size, and the number of bootstrap
repetitions on the coverage probability and width of the cheap
subsampling confidence interval.
We assess the performance of the proposed confidence interval under two scenarios:
when nuisance parameter models are estimated without employing cross-validation
and when cross-validation is used for model selection.
In the first scenario, we apply the LTMLE algorithm with two
time points after baseline, where the target parameter is the absolute
risk of an event before the second time point under sustained
treatment. We fitted parametric logistic working models without cross-validation.
Treatment and censoring models were correctly specified.
We vary
the subsample size \(m=\lfloor \eta \cdot n \rfloor\) with \(\eta \in
\{0.5, 0.632, 0.8, 0.9\}\) and the number of bootstrap repetitions \(B
\in \{5, \dots, 500\}\). We consider
sample sizes \(n \in \{250, 500, 1000, 2000, 8000\}\).
We repeat the whole procedure in 2000
simulated datasets and report aggregated results across the 2000
repetitions. Details on the data-generating mechanism, the models
used, and code for
reproducing the results are available
in the appendix (Appendix \ref{sec:datagenerating}).

We first consider the results of the simulation study without cross-validation.
In this
setting, we expect non-parametric bootstrapping with replacement to
work well. We thus include Lam's cheap bootstrap and a bootstrap
standard-error based confidence interval, both using sampling with
replacement, as competitors. In each instance, we compute the
empirical coverage and the average relative width
of our confidence
intervals for the LTMLE estimate. As another competitor we include the
asymptotic confidence intervals which are based on an estimate of the
efficient influence function \citep{laanTargetedMinimumLoss2012}. The
results are summarized in Table~\ref{table_coverage} and Figure~\ref{fig:sim_bootstrap}.

We see that the bootstrap confidence interval
based on the empirical standard error from non-parametric bootstrap
samples drawn with replacement is anticonservative when the number of bootstrap
samples is low (Figure~\ref{fig:sim_bootstrap}).
\begin{figure}[htbp]
\centering
\includegraphics[width=.85\textwidth]{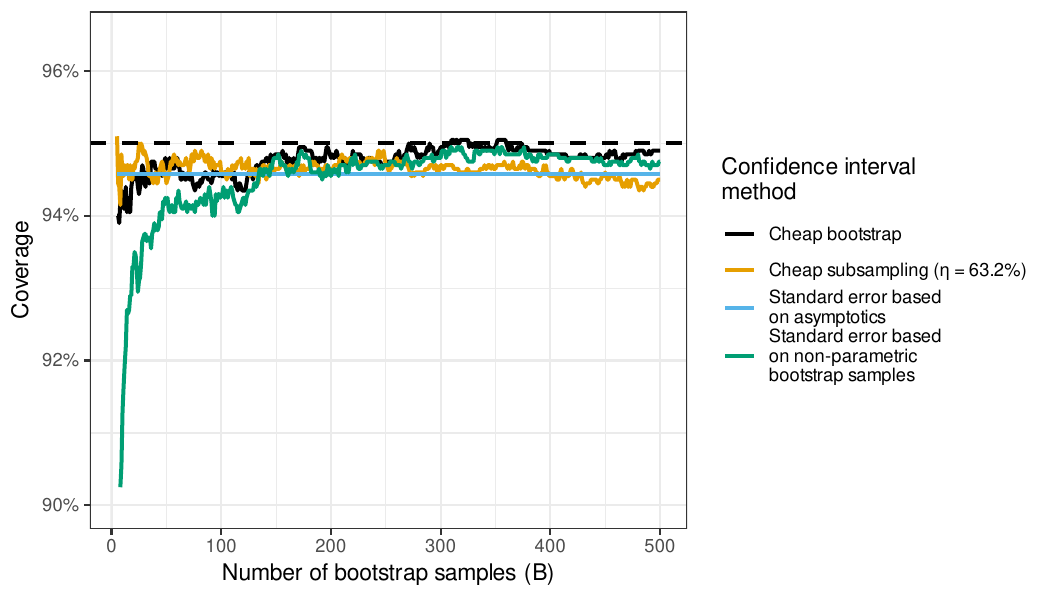}
\caption[Coverage of the 95\% cheap subsampling confidence interval
when cross-validation is not applied against number of bootstrap repetitions
(\(m=\lfloor 0.632 \cdot 2000 \rfloor = 1264\))]{
  \label{fig:sim_bootstrap}
  Results from the first simulation study using working
  parametric models without cross-validation.
  The coverage is shown on the \( y \)-axis for
  four 95\% confidence intervals for the absolute risk of an event
  before the second time point under sustained exposure
  using the LTMLE in a sample size of \(n=2000\). The four confidence
  intervals are the asymptotic confidence interval, the cheap
  subsampling confidence interval
  (\(m=\lfloor 0.632 \cdot 2000 \rfloor = 1264\)), the cheap
  bootstrap confidence interval
  \citep{lamCheapBootstrapMethod2022}, and a confidence interval
  based on non-parametric bootstrap standard error. The \( x \)-axis shows the
  number of bootstrap repetitions \(B \in \{5, \dots, 500\}\).}
\end{figure}
\begin{figure}[htbp]
\centering
\includegraphics[width=.85\textwidth]{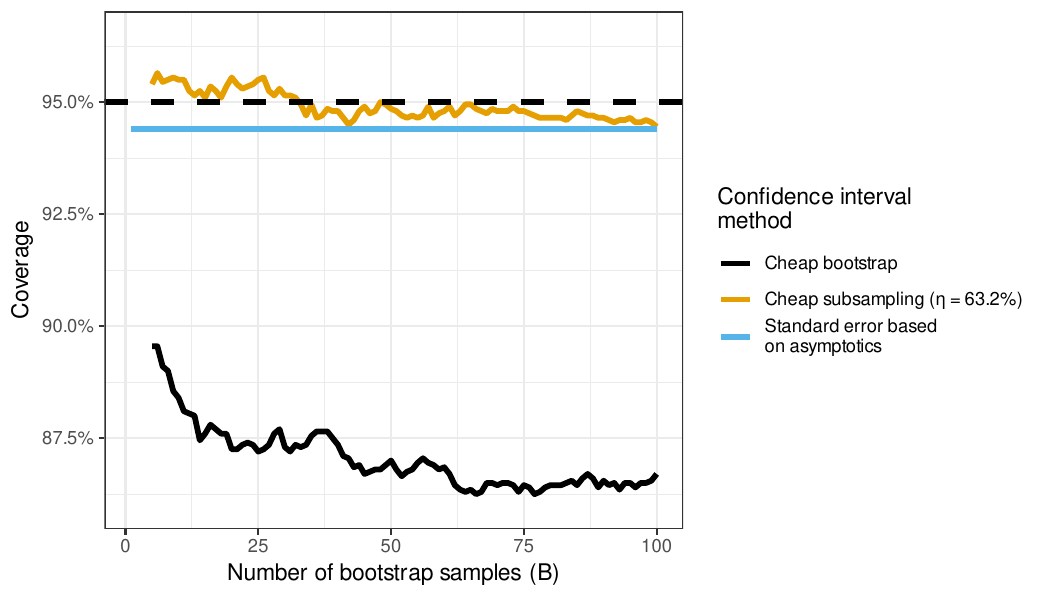}
\caption[Coverage of the 95\% cheap subsampling confidence interval
when cross-validation is applied against number of bootstrap repetitions
(\(m=\lfloor 0.632 \cdot 1000 \rfloor = 632\))]{\label{fig:sim_bootstrap_cv}
  Results from the second simulation study where cross-validation is
  applied as part of the nuisance parameter estimation step. The
  coverage is shown on the \( y \)-axis for three 95\% confidence
  intervals for the absolute risk of an event before the first time
  point under sustained exposure using the LTMLE in a sample size
  of \(n=1000\). The three confidence intervals are
  the asymptotic confidence interval,
  the cheap subsampling confidence interval
  (\(m=\lfloor 0.632 \cdot 1000 \rfloor = 632\)), and the cheap
  bootstrap confidence interval
  \citep{lamCheapBootstrapMethod2022}. The \( x \)-axis shows the number of
  bootstrap repetitions \(B \in \{5, \dots, 100\}\).}
\end{figure}

Figure~\ref{fig:sim_bootstrap} also shows that the coverages of the
cheap subsampling and the cheap bootstrap confidence intervals are
close to the nominal level already with few bootstrap repetitions.
When we compare the coverage of the asymptotic confidence interval
with the coverage of the cheap bootstrap confidence interval
\citep{lamCheapBootstrapMethod2022}, we see that the cheap subsampling
confidence interval has similar coverage.

Table \ref{table_coverage} shows no systematic effects on the coverage
of the cheap subsampling confidence interval for different subsample
sizes and the number of bootstrap repetitions
in large sample sizes. However, the coverage appears to depend
on the subsample size and
number of bootstrap repetitions
when the sample size is small. For the sample size
\(n=250\), we see that a small subsample size \(m=\lfloor 0.5 \cdot 250
\rfloor = 125\) performs well with \(B=5\). Additional results for the other
subsample sizes \(\eta \in \{0.5, 0.8, 0.9\}\) as
well as runtime comparisons are included in
the appendix (Appendix \ref{sec:additional_results};
Table \ref{tab:runtime} and Figure \ref{fig:sim_bootstrap_others}).

Table \ref{table_coverage} summarizes the relative widths of our
confidence intervals compared to the asymptotic confidence interval.
We see that the cheap subsampling confidence
interval is, in general, slightly wider than the asymptotic confidence
interval, but that increasing \(B\) leads to a relative width
approaching 100\%, with
little dependence on the subsample
fraction in larger samples. Similar results were also
observed for the cheap bootstrap confidence interval
\citep{pmlr-v202-lam23a, cheapbootstrap2}. A possible explanation for
the wider cheap subsampling confidence intervals seen for small
numbers of bootstrap replications is that the quantiles of the
\(t\)-distribution are large for low degrees of freedom but quite
comparable to the normal distribution for large degrees of freedom
(\(B \geq 25\)). Note that we have seen a similar behaviour of the
cheap subsampling confidence intervals in our case study (Section
\ref{example_ltmle}). Moreover, we note that the relative width of the
cheap subsampling confidence interval decreases with increasing
subsample size and with increasing sample size, but this effect is
less pronounced than the effect of the number of bootstrap repetitions.

\begin{table*}
  \caption[Coverage and relative widths (compared to asymptotic confidence interval) of the 95\% cheap subsampling confidence interval compared to the asymptotic confidence interval]{\label{table_coverage}
The table shows the coverage and relative widths compared to the asymptotic confidence interval in the first simulation setting without cross-validation of the 95\% cheap subsampling confidence interval for the absolute risk of an event before the end of the second time interval under sustained exposure using the LTMLE for different subsample fractions \(\eta\) and sample sizes \(n\) and the number of bootstrap repetitions \(B\). 
}
\centering
\begin{tabular}{rr|rrrr|rrrr} \toprule
& \multicolumn{1}{c}{} & \multicolumn{4}{c}{Coverage (\%)} & \multicolumn{4}{c}{Relative width (\%)}  \\
& & \multicolumn{4}{c|}{Subsample proportion (\(\eta\))} & \multicolumn{4}{c}{Subsample proportion (\(\eta\))} \\
\(B\) & \(n\) & 50\% & 63.2\% & 80\% & 90\% & 50\% & 63.2\% & 80\% & 90\% \\
  \midrule
   5 & 250 & 94.8 & 93.2 & 92.9 & 93.7 & 131.0 & 127.5 & 127.5 & 127.2 \\
    & 500 & 94.2 & 93.8 & 95.0 & 94.5 & 126.3 & 126.0 & 126.9 & 125.8 \\
    & 1000 & 94.0 & 95.0 & 94.2 & 94.7 & 125.1 & 125.8 & 124.9 & 125.3 \\
    & 2000 & 94.8 & 95.1 & 95.0 & 94.5 & 125.6 & 126.2 & 123.5 & 124.8 \\
    & 8000 & 95.3 & 95.8 & 94.5 & 95.3 & 125.8 & 127.1 & 124.8 & 124.0 \\
   \midrule
  25 & 250 & 93.6 & 92.5 & 93.2 & 92.2 & 107.9 & 106.0 & 106.1 & 106.3 \\
   & 500 & 94.2 & 93.8 & 94.3 & 95.1 & 105.8 & 104.9 & 104.8 & 104.1 \\
   & 1000 & 94.5 & 95.2 & 94.2 & 94.3 & 104.3 & 104.8 & 104.5 & 104.1 \\
   & 2000 & 94.2 & 95.0 & 95.3 & 95.0 & 104.9 & 105.2 & 104.2 & 104.3 \\
   & 8000 & 94.7 & 95.3 & 95.3 & 94.5 & 104.2 & 104.8 & 104.0 & 103.6 \\
   \midrule
  100 & 250 & 93.8 & 92.7 & 93.2 & 92.5 & 104.8 & 103.5 & 103.2 & 103.0 \\
   & 500 & 94.2 & 93.8 & 94.9 & 95.0 & 102.5 & 102.2 & 101.8 & 101.7 \\
   & 1000 & 94.9 & 94.8 & 94.6 & 93.8 & 101.5 & 101.5 & 101.4 & 101.1 \\
   & 2000 & 94.2 & 94.6 & 94.8 & 95.0 & 101.5 & 101.3 & 101.0 & 101.1 \\
   & 8000 & 94.5 & 95.2 & 95.0 & 95.4 & 101.0 & 101.3 & 100.7 & 101.0 \\
   \midrule
  500 & 250 & 93.9 & 92.8 & 93.0 & 92.7 & 103.9 & 102.8 & 102.2 & 102.1 \\
   & 500 & 93.8 & 93.8 & 94.7 & 95.1 & 101.7 & 101.3 & 101.1 & 101.0 \\
   & 1000 & 94.8 & 94.8 & 94.5 & 94.0 & 100.8 & 100.7 & 100.6 & 100.5 \\
   & 2000 & 94.0 & 94.5 & 94.8 & 95.0 & 100.6 & 100.5 & 100.3 & 100.3 \\
   & 8000 & 94.6 & 94.9 & 95.0 & 94.7 & 100.2 & 100.3 & 100.2 & 100.2 \\
    \bottomrule
\end{tabular}
\end{table*}

In the scenario with cross-validation, we apply the LTMLE
algorithm with one time point after baseline, where the target
parameter is the absolute risk of an event before the first time point
under treatment at baseline.
Note that we can apply LTMLE as it handles the point exposure setting
as a special case.
Here, we fix the sample size \(n=1000\)
but vary \(m=\lfloor \eta \cdot n \rfloor\) with \(\eta \in
\{0.5, 0.632, 0.8, 0.9\}\)
and \(B \in \{5, \dots, 100\}\)
and repeat the simulations 2000 times.
We investigate the
effect of employing cross-validation in the nuisance parameter
estimation step by using 10-fold cross-validation to select between a
working logistic regression model and a random
forest \citep{breiman2001random}.

Figure~\ref{fig:sim_bootstrap_cv} shows clearly that the
cheap bootstrap confidence interval where bootstrap samples are drawn
with replacement has poor coverage which is decreasing with
the number of subsampling repetitions.
This highlights the issues
found in Figure~\ref{fig:bootstrap_problem} and Figure~\ref{fig:bootstrap_problem_example}
with using cross-validation and resampling with replacement.
We also find that our cheap subsampling confidence interval has
coverage reasonably close to the nominal level which is comparable to that of the
asymptotic confidence interval. An additional figure for other
subsample sizes \(\eta \in \{0.5, 0.8, 0.9\}\) is included in the
appendix (Appendix \ref{sec:additional_results}; Figure \ref{fig:sim_bootstrap_others_cv}).

\section{Discussion}

We introduced a cheap subsampling bootstrap confidence interval for
asymptotically linear estimators that is computationally efficient and
avoids resampling with replacement. The method is particularly well
suited to settings in which nuisance parameters are estimated using
cross-validation or cross-fitting, where conventional bootstrap
procedures may perform poorly.

Across simulation studies, the proposed confidence interval
exhibited coverage close to
the nominal level and performance comparable to asymptotic confidence
intervals for large sample sizes.
While the method is asymptotically valid for any number of
subsamples, our empirical results indicate that using too few
resamples may lead to noticeable Monte Carlo variability, underscoring
the practical importance of choosing a moderate number of subsamples.

As \cite{lamCheapBootstrapMethod2022} showed
for his cheap bootstrap method, the
expected width of the confidence interval is large
if we choose a small number of subsamples.
By the same arguments that are present in \cite{lamCheapBootstrapMethod2022},
we can therefore construct expressions
for the (asymptotic) expected width of the cheap subsampling confidence interval
that would be identical to his.
In fact, our results are in alignment with his tables
on the expected width.
We recommend using at least 25 subsamples to get
a relative width of the confidence interval which would
be around 5\% larger than the asymptotic confidence interval.

The finite-sample performance was more sensitive to the choice of
subsample size in smaller samples, whereas this dependence diminished
with increasing sample size.
For large datasets, we recommend
choosing a subsample size of around 63\% of the original sample size
to save on computation time
as a pragmatic choice.
For smaller datasets, the recommendation
is less clear, but also less motivated
for our use case.
A problem may be whether the nuisance parameter
models can actually be estimated with a smaller subsample size.
Principled data-adaptive selection of the subsample size
remains an open area for future research.

For longitudinal targeted minimum loss-based estimation (LTMLE)
\citep{laanTargetedMinimumLoss2012}, one may also use bootstrap algorithms that are
specifically designed for it
\citep{tranRobustVarianceEstimation2023,Coyle2018}.
\citet{tang2024consistencybootstrapasymptoticallylinear}
consider a smoothed bootstrap procedure for asymptotically linear estimators
using cross-validation for nuisance parameter estimation.
It would be interesting to compare the performance of these methods
with the cheap subsampling confidence interval.

In summary, the cheap subsampling confidence interval offers a simple
and computationally attractive alternative to standard bootstrap
methods for inference, particularly in modern
applications involving flexible nuisance parameter estimation.

\section*{Acknowledgements}
The authors would like to thank Novo Nordisk for providing
 the data from the LEADER trial.

\section*{Conflicts of interest}
None declared. 

\section*{Funding}
This work is part of the REDDIE study and has been partially funded by the European Union. 
Views and opinions expressed are however those of the author(s) only and do not necessarily reflect those of the European Union or European Health and Digital Executive Agency (HADEA).
Neither the European Union nor the granting authority can be held responsible for them. This work has received funding from the UK Research and Innovation under contract number 101095556.

\section*{Data availability}
The datasets analyzed during the current study are not publicly available due
to the proprietary nature of the data.
Data can be obtained upon request from Novo Nordisk.

\begingroup
\raggedright
\bibliographystyle{apalike}
\bibliography{references}
\endgroup

\begin{appendix}

\section{Proof of Theorem 1}
\label{pf_thm1}
To prove the theorem, we use the notation and framework of Section
\ref{cheap_sub}.  Since the estimator \(\est\) is asymptotically linear,
Slutsky's theorem and the central limit theorem yield
\begin{equation*}
\sqrt{n} (\est-\estimand) \weakly Z \distequal \normal 
\end{equation*}
where \(\sigma^2 = \mathbb{E}_P[\phi_P(O)^2] > 0\). 
Theorem 2 (iii) of \cite{wuJackknife1990} gives that
\begin{equation*}
\sup_{x \in \mathbb{R}} \left|P\left(\subsamplingfactor(\bootest-\est) \leq x | \data \right) - \Phi_{\sigma^2}(x) \right| \inprob 0, 
\end{equation*}
if \(\frac{n-m}{n} > \lambda\) for some \(\lambda > 0\) for all \(n \in \mathbb{N}\).
By our assumption on the subsample size, we have \(\subpercentage < c\) for some \(0<c<1\), and thus \(\frac{n-m}{n}  >\lambda := 1-c > 0\).
This implies
\begin{equation*}
P\left(\subsamplingfactor(\bootest-\est) \leq x | \data \right) \inprob \Phi_{\sigma^2}(x),  
\end{equation*}
for all \(x \in \mathbb{R}\) as \(m, n \rightarrow \infty\), where
\(\Phi_{\sigma^2}\) is the cumulative distribution function of a
Normal distribution with mean 0 and variance \(\sigma^2\). The
remainder of the proof follows along the steps of the proof of Theorem
1 in \cite{lamCheapBootstrapMethod2022}. Let \(k_{(m,n)} =
\frac{mn}{n-m}\). We want to show that
\begin{align}
&(\sqrt{n}(\est-\estimand), \sqrt{k_{(m,n)}}(\hat{\Psi}^*_{(m,1)}- \est), \dots, \sqrt{k_{(m,n)}}(\hat{\Psi}^*_{(m,B)}- \est)) \tag*{} \\
&\weakly (Z_0, \dots, Z_B), \label{eq:jointbootstrap}
\end{align}
where \(Z_0, \dots, Z_B\) are independent and identically distributed with \(Z_b \distequal Z\).
By conditioning on \(\data\) and using conditional independence of \(\sqrt{k_{(m,n)}}(\hat{\Psi}_{(m,b)}^*- \est), b = 1, \dots, B\)  (bootstrap samples are drawn independently given the data), we have
\begin{align}
	&\Big|P\left(\sqrt{n}(\est-\estimand) \leq z_0, \right. \tag*{} \\
	&\left. \sqrt{k_{(m,n)}}(\hat{\Psi}^*_{(m,1)}-\est)\leq z_1, \dots, \sqrt{k_{(m,n)}}(\hat{\Psi}^*_{(m,B)}-\est) \leq z_B\right) -\prod_{b=0}^{B}\Phi_{\sigma^2}(z_b) \Big| \tag*{}\\
	&\leq \mathbb{E}\left[I\left(\sqrt{n}(\est-\estimand) \leq z_0\right) \right. \tag*{} \\
	&\left. \left| \prod_{b=1}^B P\left(\sqrt{k_{(m,n)}}(\hat{\Psi}^*_{(m,b)}-\est)\leq z_b | \data\right) -\prod_{b=1}^B \Phi_{\sigma^2}(z_b)\right|\right] \label{eq:jointbootstrapeq1} \\
	&+ \left|\prod_{b=1}^B \Phi_{\sigma^2}(z_b) \left[P\left(\sqrt{n}(\est-\estimand)\leq z_0\right)  - \Phi_{\sigma^2}(z_0)\right]\right|, \label{eq:jointbootstrapeq2}
\end{align}
for any \(z=(z_0, \dots, z_B) \in \mathbb{R}^{B+1}\).  Since, by
assumption, \(\sqrt{n} (\hat\Psi_n-\estimand) \weakly Z\), the term in
equation \eqref{eq:jointbootstrapeq2} converges to zero as \(n
\rightarrow \infty\).  Since also
\(P(\sqrt{k_{(m,n)}}(\hat{\Psi}^*_{(m,b)} - \hat\Psi_n) \leq z |
\data) \inprob \Phi_{\sigma^2}(z)\) for \(b=1, \dots, B\) and all \(z
\in \mathbb{R}\) as \(n \rightarrow \infty\), it follows that the
integrand of the term \eqref{eq:jointbootstrapeq1} converges to zero in
probability as \(n \rightarrow \infty\). Since the integrand in the
term \eqref{eq:jointbootstrapeq1} is bounded by 1, it follows from
dominated convergence that (\ref{eq:jointbootstrapeq1}) tends to
zero. Thus \eqref{eq:jointbootstrap} holds. From this result, we deduce
that
\begin{align*}
T_{(m,n)} &= \frac{\est-\estimand}{\sqrt{\frac{m}{n-m}}S} = \frac{\sqrt{n}(\est-\estimand)}{\sqrt{k_{(m,n)}}S} \\
&= \frac{\frac{\sqrt{n}(\est-\estimand)}{\sigma}}{\sqrt{\frac{1}{B} \sum_{b=1}^B \left(\frac{\sqrt{k_{(m,n)}}}{\sigma} (\hat{\Psi}^*_{(m,b)}-\est)\right)^2}} \weakly \frac{\tilde{Z}_0}{\sqrt{\frac{1}{B} \sum_{b=1}^B \tilde{Z}_b^2}}
\end{align*}
where \(\tilde{Z}_b = Z_b/\sigma \distequal N(0, 1)\).
Note that by the independence of the \(\tilde{Z}_b\)'s and the fact that \(\tilde{Z}_b \distequal N(0,1)\), we have that
\(\frac{\tilde{Z}_0}{\sqrt{\frac{1}{B} \sum_{b=1}^B \tilde{Z}_b^2}}\) has a \(t\)-distribution with \(B\) degrees of freedom.
This shows that \(T_{(m,n)}\) converges in distribution to a \(t\)-distribution with \(B\) degrees of freedom,
as \(n \rightarrow \infty\). Thus, we have as \(n \rightarrow \infty\),
\begin{equation*}
P( \Psi(P) \in \csint) = P(-t_{B,1-\alpha/2} < T_{(m,n)} < t_{B,1-\alpha/2}) \rightarrow 1-\alpha.
\end{equation*}

\section{Data-generating mechanism}
\label{sec:datagenerating}

In this section, we describe the data-generating mechanism for our
simulation study (Section 4). See
\url{https://github.com/jsohlendorff/cheap\_subsampling\_simulation\_study}
for the \textbf{R} code.
First, we describe the setting with two time points.
We denote by \(Y_t\) the outcome, where \(Y_t=1\) if an
event has occurred by time \(t\) and \(Y_t=0\) otherwise.  We also
denote the treatment by \(A_t\) (1 is treated and 0 is untreated), the
time-varying confounders by \(W_t\) (continuous), and the censoring
indicators by \(C_t\) (1 is uncensored and 0 is censored). We generate
the data at baseline (\(t = 0\)) and at two time points (\(t = 1\) and
\(t = 2\)) as follows:
\begin{align*}
  W_0 &\sim \mathcal{N}(0,1), \\
  A_0 &\sim \text{Bern}(\expit(0.4 - 0.2 W_0)), \\
  C_1 &\sim \text{Bern}(\expit(3.5 + W_0)), \\
  Y_1 &\sim \begin{cases}
    \text{Bern}(\expit(-1.4 + 0.1 W_0 - 1.5 A_0)), & \text{if } C_1 = 1 \\
    \emptyset, & \text{if } C_1 = 0,
  \end{cases} \\
  W_1 &\sim \begin{cases}
    \mathcal{N}(0, 1), & \text{if } C_1 = 1 \text{ and } Y_1 = 0 \\
    \emptyset, & \text{otherwise},
  \end{cases} \\
  A_1 &\sim \begin{cases}
    \text{Bern}(\expit(-0.4 W_0 + 0.8 A_0)), & \text{if } C_1 = 1 \text{ and } Y_1 = 0 \\
    \emptyset, & \text{otherwise},
  \end{cases} \\
  C_2 &\sim \begin{cases}
    \text{Bern}(\expit(3.5 + W_1)), & \text{if } C_1 = 1 \text{ and } Y_1 = 0 \\
    0 & \text{if } C_1 = 0 \\
    \emptyset, & \text{otherwise},
  \end{cases} \\
  Y_2 &\sim \begin{cases}
    \text{Bern}(\expit(-1.4 + 0.1 W_1 - 1.5 A_1)), & \text{if } C_2 = 1 \text{ and } Y_1 = 0 \\
    1 & \text{if } Y_1 = 1 \\
    \emptyset, & \text{otherwise},
  \end{cases}
\end{align*}
where \(\expit(x) = \frac{1}{1 + \exp(-x)}\) is the logistic function, \(\mathcal{N}(\mu, \sigma^2)\) is the normal distribution with mean \(\mu\) and variance \(\sigma^2\),
\(\text{Bern}(p)\) is the Bernoulli distribution with probability \(p\), and \(\emptyset\) is the missingness indicator.

Next, we describe the setting with one time point used
to conduct the simulations with cross-validation.
Here, we simulate as follows:
\begin{align*}
  W_0 &\sim \mathcal{N}(0,1), \\
  A_0 &\sim \text{Bern}(\expit(0.5 W_0)), \\
  C_1 &= 1, \\
  Y_1 &\sim \begin{cases}
    \text{Bern}(\expit(0.2 W_0 - 0.5 A_0)), & \text{if } C_1 = 1 \\
    \emptyset, & \text{if } C_1 = 0.
  \end{cases} 
\end{align*}
We apply the TMLE procedure using 10-fold cross-validation
where we select between a working logistic regression model
and a random forest \citep{breiman2001random} using
the \verb|ranger| package \citep{ranger}.
We specified \verb|max.depth = 60| and \verb|min.node.size = 300|
in the \verb|ranger| call
and used the function's defaults otherwise.

\section{Additional simulation results}
\label{sec:additional_results}

\begin{figure} 
\centering
\includegraphics[width=.85\textwidth]{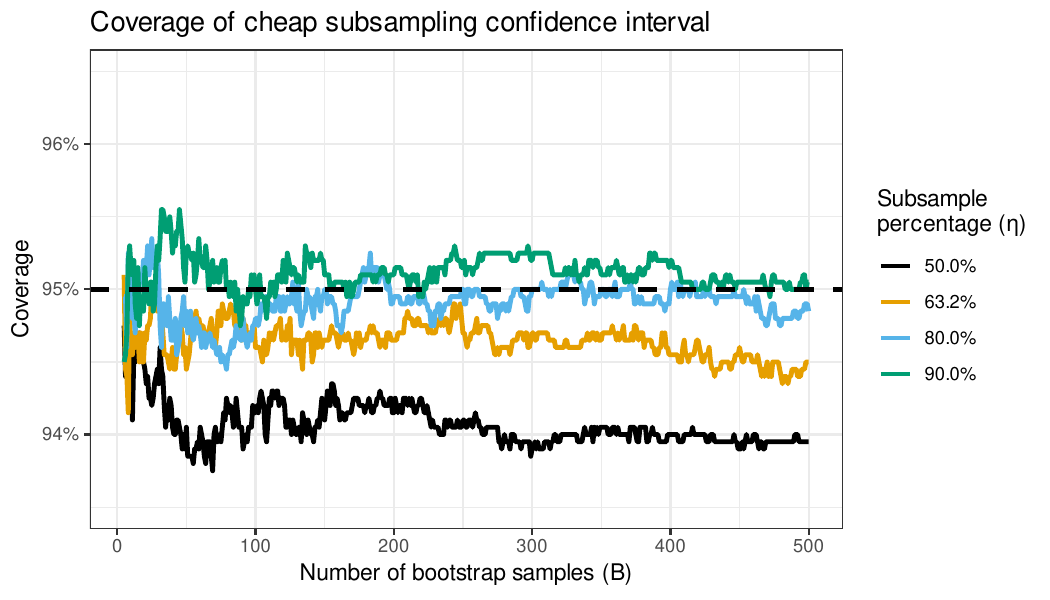}
\caption[Coverage of the 95\% cheap subsampling confidence interval
when cross-validation is not applied against number of bootstrap repetitions
for various subsample sizes)]{\label{fig:sim_bootstrap_others}
  Results from the simulation study for the setting without cross-validation illustrating the coverage (y-axis) of
  95\% cheap subsampling confidence intervals for the absolute risk of an event before the end 
  of the second time interval under sustained exposure using LTMLE for \(n=2000\)
  for the subsample sizes \(m = \lfloor \eta \cdot n \rfloor\),
  where \(\eta \in \{0.5,0.632,0.8,0.9\}\). The \( x \)-axis shows the
  number of bootstrap repetitions \(B \in \{5, \dots, 500\}\).}
\end{figure} 

\begin{figure}
\centering
\includegraphics[width=.85\textwidth]{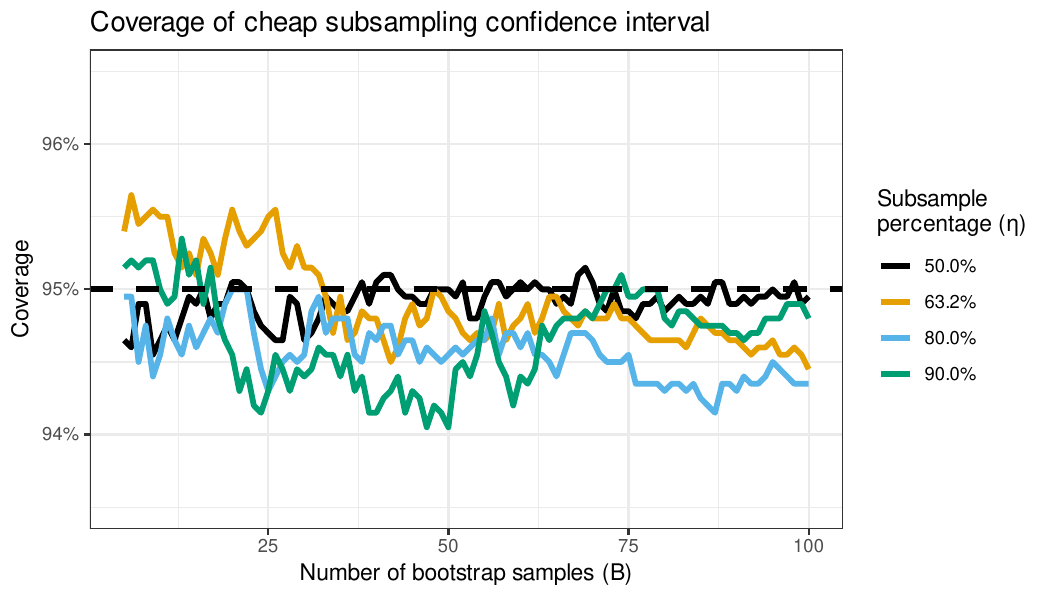}
\caption[Coverage of the 95\% cheap subsampling confidence interval
  when cross-validation is applied against number of bootstrap repetitions 
  for various subsample sizes)]{\label{fig:sim_bootstrap_others_cv}
  Results from the simulation study for the setting with cross-validation illustrating the coverage (y-axis) of
  95\% cheap subsampling confidence intervals for the absolute risk of an event before the end 
  of the first time interval under sustained exposure using LTMLE for \(n=1000\)
  for the subsample sizes \(m = \lfloor \eta \cdot n \rfloor\),
  where \(\eta \in \{0.5,0.632,0.8,0.9\}\). The \( x \)-axis shows the
  number of bootstrap repetitions \(B \in \{5, \dots, 100\}\).
}
\end{figure}

\begin{table*}
  \centering
  \caption{The table shows the median (25th, 75th percentiles) across 25 repetitions
    for the runtime in seconds for our cheap subsampling confidence interval with $m$
    across \(\eta = m/n \in \{0.5,0.632,0.8,0.9\}\)
    and the cheap bootstrap confidence interval \citep{lamCheapBootstrapMethod2022}
    with replacement
    with a single simulated dataset with \(n=10,000\) observations
    (same setup as our simulations without cross-validation).
    Bootstrap replications within each simulation were run sequentially.
    The simulations were run on a computer
    with an AMD Ryzen 7 8845HS processor with 8 cores, 16 threads, and 32 GB of RAM. } \label{tab:runtime}
  \begin{tabular}{@{}rrrc@{}}
    \toprule
    $B$ & $\eta$ & Replacement
        & Median runtime (Q1, Q3) \\
    \midrule
     5 & 50\% & No  & 1.14 (1.12, 1.32) \\
      & 63.2\% &   & 1.29 (1.25, 1.48) \\
      & 80\% &   & 1.60 (1.42, 1.65) \\
      & 90\% &   & 1.76 (1.56, 1.80) \\
      & 100\% & Yes & 1.96 (1.90, 2.15) \\
    \addlinespace
    25 & 50\% & No  & 5.17 (5.06, 5.27) \\
     & 63.2\% &   & 6.22 (6.07, 6.36) \\
     & 80\% &   & 7.64 (7.52, 7.90) \\
     & 90\% &   & 8.33 (8.07, 8.53) \\
     & 100\% & Yes & 9.65 (9.46, 9.80) \\
    \bottomrule
  \end{tabular}
\end{table*}

\end{appendix}

\end{document}